\documentclass[11pt]{article}

\usepackage{algorithm}
\usepackage{algorithmic}
\usepackage{framed}
\usepackage{scalerel}
\usepackage{mathrsfs}
\usepackage{enumitem}
\usepackage{bbm}
\usepackage{tikz}
\usepackage{float}
\usetikzlibrary{backgrounds}
\usepackage{color}
\usepackage{graphicx}
\usepackage{latexsym}
\usepackage{amsfonts}
\usepackage{pifont,xspace,fullpage,epsfig, wrapfig}
\usepackage{amsmath, amssymb, amsthm} %
\usepackage{multirow}
\usepackage{dsfont}
\usepackage{array}
\usepackage[bookmarks=false]{hyperref}
\usepackage{adjustbox}

\usepackage{bookmark}

\DeclareSymbolFont{AMSb}{U}{msb}{m}{n}
\DeclareMathSymbol{\N}{\mathbin}{AMSb}{"4E}
\DeclareMathSymbol{\Z}{\mathbin}{AMSb}{"5A}
\DeclareMathSymbol{\R}{\mathbin}{AMSb}{"52}
\DeclareMathSymbol{\Q}{\mathbin}{AMSb}{"51}
\DeclareMathSymbol{\erert}{\mathbin}{AMSb}{"50}
\DeclareMathSymbol{\I}{\mathbin}{AMSb}{"49}
\DeclareMathSymbol{\C}{\mathbin}{AMSb}{"43}
\newcommand{\E}{\mathbb{E}}

\newcommand{\mynote}[2]{{\textcolor{#1}{ #2}}}

\newcommand{\red}[1]{\mynote{red}{#1}}
\definecolor{gray}{gray}{0.4}

\newcommand{\remove}[1]{}

\newtheorem{theorem}{Theorem}[section]
\newtheorem{lemma}[theorem]{Lemma}
\newtheorem{definition}[theorem]{Definition}
\newtheorem{remark}[theorem]{Remark}

\newtheorem{claim}[theorem]{Claim}

\newtheorem{observation}[theorem]{Observation}
\newtheorem{question}[theorem]{Question}

\newcommand{\AAA}{\mathcal A}
\newcommand{\BBB}{\mathcal B}
\newcommand{\oracle}{\mathcal{E}}

\newcommand{\DDD}{\mathcal D}

\newcommand{\eps}{\varepsilon}

\newcommand{\Gen}{{\rm Gen}}
\newcommand{\Enc}{{\rm Enc}}
\newcommand{\Dec}{{\rm Dec}}

\newcommand{\PRG}{\operatorname{\rm PRG}}
\newcommand{\SAMP}{\operatorname{\rm SAMP}}

\newcommand{\polylog}{\mathop{\rm polylog}}

\newcommand{\poly}{\mathop{\rm poly}}

\newcommand{\negl}{\mathop{\rm negl}}

\def\E{\operatorname*{\mathbb{E}}}

\def\Q{\operatorname*{\mathbb{Q}}}
\def\poly{\mathop{\rm{poly}}\nolimits}

\newcommand{\cM}{\mathcal{M}}

\newcommand{\accgame}{\mathsf{Acc}}

\makeatletter
\newcommand{\thickhline}{%
    \noalign {\ifnum 0=`}\fi \hrule height 1pt
    \futurelet \reserved@a \@xhline
}
\newcolumntype{"}{@{\hskip\tabcolsep\vrule width 1pt\hskip\tabcolsep}}
\makeatother

\makeatletter
\newlength{\fboxhsep}
\newlength{\fboxvsep}
\newlength{\fboxtoprule}
\newlength{\fboxbottomrule}
\newlength{\fboxleftrule}
\newlength{\fboxrightrule}
\def\@frameb@xother#1{%
  \@tempdima\fboxtoprule
  \advance\@tempdima\fboxvsep
  \advance\@tempdima\dp\@tempboxa
  \hbox{%
    \lower\@tempdima\hbox{%
      \vbox{%
        \hrule\@height\fboxtoprule
        \hbox{%
          \vrule\@width\fboxleftrule
          #1%
          \vbox{%
            \vskip\fboxvsep
            \box\@tempboxa
            \vskip\fboxvsep}%
          #1%
          \vrule\@width\fboxrightrule}%
        \hrule\@height\fboxbottomrule}%
    }%
  }%
}
\long\def\fboxother#1{%
  \leavevmode
  \setbox\@tempboxa\hbox{%
    \color@begingroup
    \kern\fboxhsep{#1}\kern\fboxhsep
    \color@endgroup}%
  \@frameb@xother\relax}

\makeatother

\begin{document}

\begin{titlepage}

\title{Separating Adaptive Streaming from Oblivious Streaming}
\author{
Haim Kaplan\thanks{Tel Aviv University and Google Research. \texttt{haimk@tau.ac.il}. Partially supported by the Israel Science Foundation (grant 1595/19), the German-Israeli Foundation (grant 1367/2017), and by the Blavatnik Family Foundation.}
\and
Yishay Mansour\thanks{Tel Aviv University and Google Research. \texttt{mansour.yishay@gmail.com}. This project has received funding from the European Research Council (ERC) under the European Union's Horizon 2020 research and innovation program (grant agreement No. 882396), and by the Israel Science Foundation (grant number 993/17).}
\and
Kobbi Nissim\thanks{Georgetown University. \texttt{kobbi.nissim@georgetown.edu}. Supported by NSF grant No. 1565387 TWC: Large: Collaborative: Computing Over Distributed Sensitive Data.}
\and
Uri Stemmer\thanks{Ben-Gurion University and Google Research. \texttt{u@uri.co.il}. Partially Supported by the Israel Science Foundation (grant 1871/19) and by the Cyber Security Research Center at Ben-Gurion University of the Negev.}
}

\date{February 17, 2021}
\maketitle

\begin{abstract}
We present a streaming problem for which every adversarially-robust streaming algorithm must use polynomial space, while there exists a classical (oblivious) streaming algorithm that uses only polylogarithmic space. This is the first separation between oblivious streaming and adversarially-robust streaming, and resolves one of the central open questions in adversarial robust streaming.
\end{abstract}

\end{titlepage}

\section{Introduction}
Consider a scenario in which data items are being generated one by one, e.g., IP traffic monitoring or web searches. 
Generally speaking, streaming algorithms aim to process such data streams while using only a limited amount of memory, significantly smaller than what is needed to store the entire data stream. Streaming algorithms have become a central and crucial tool for the analysis of massive datasets.

A typical assumption when designing and analyzing streaming algorithms is that the entire stream is {\em fixed} in advance (and is just provided to the streaming algorithm one item at a time), or at least that the choice of the items in the stream is {\em independent} of the internal state (and coin tosses) of the streaming algorithm. %
We refer to this setting as the {\em oblivious} setting. 
Recently, there has been a growing interest in streaming algorithms that maintain utility even when the choice of stream items depends on previous answers given by the streaming algorithm, and can hence depend on the internal state of the algorithm~\cite{MironovNS11,GHRSW12,GHSWW12,AhnGM12,AhnGM12b,HardtW13,BenEliezerY19,BenEliezerJWY20,HassidimKMMS20,WoodruffZhou20}. Such streaming algorithms are said to be {\em adversarially robust}.

Hardt and Woodruff~\cite{HardtW13} presented a negative result showing that, generally speaking, {\em linear} algorithms cannot be adversarially robust.\footnote{A streaming algorithm is {\em linear} if for some (possibly randomized) matrix $A$, its output depends only on $A$ and $Af$, where $f$ is the {\em frequency vector} of the stream.} This result does not rule out non-linear algorithms. Indeed, strong positive results were shown by~\cite{BenEliezerJWY20,HassidimKMMS20,WoodruffZhou20} who constructed (non-linear) adversarially robust algorithms for many problems of interest, with small overhead compared to the oblivious setting. 
This includes problems such as estimating frequency moments, counting the number of distinct elements in the stream, identifying heavy-hitters in the stream, estimating the median of the stream, entropy estimation, and more.
The strong positive results of~\cite{BenEliezerJWY20,HassidimKMMS20,WoodruffZhou20} raise the possibility that adversarial robustness can come ``for free'' in terms of the additional costs to memory, compared to what is needed in the oblivious setting. 

\begin{question}
Is there a gap between adversarial streaming and oblivious streaming? %
\end{question}

We provide a positive answer to this question. Specifically, we present a streaming problem for which every adversarially-robust streaming algorithm must use polynomial space, while there exists an oblivious streaming algorithm that uses only polylogarithmic space. 

\subsection{Streaming against adaptive adversaries}
Before describing our new results, we define our setting more precisely.
A stream of length $m$ over a domain $X$ consists of a sequence of updates $x_1,\dots,x_m\in X$. For $i\in[m]$ we write $\vec{x}_i=(x_1,\dots,x_i)$ to denote the first $i$ updates of the stream. 
Let $g:X^*\rightarrow\R$ be a function (for example, $g$ might count the number of distinct elements in the stream).
At every time step $i$, after obtaining the next element in the stream $x_i$, our goal is to output an approximation for $g(\vec{x}_i)$. 
Throughout the paper we use $\alpha$ for the approximation parameter and $\beta$ for the confidence parameter.

The adversarial streaming model, in various forms, was considered by~\cite{MironovNS11,GHRSW12,GHSWW12,AhnGM12,AhnGM12b,HardtW13,BenEliezerY19,BenEliezerJWY20,HassidimKMMS20,WoodruffZhou20}. We give here the formulation presented by Ben-Eliezer et al.~\cite{BenEliezerJWY20}. 
The adversarial setting is modeled by a two-player game between a (randomized) \texttt{StreamingAlgorithm} and an \texttt{Adversary}. At the beginning, we fix a function $g$. Then the game proceeds in rounds, where in the $i$th round:

\begin{enumerate}
	\item The \texttt{Adversary} chooses an update $x_i\in X$ for the stream, which can depend, in particular, on all previous stream updates and outputs of \texttt{StreamingAlgorithm}.
	\item The \texttt{StreamingAlgorithm} processes the new update $x_i$ and outputs its current response $z_i$.
\end{enumerate}

The goal of the \texttt{Adversary} is to make the \texttt{StreamingAlgorithm} output an incorrect response $z_i$ at some point $i$ in the stream, that is $z_i\notin(1\pm\alpha)\cdot g(\vec{x}_i)$. For example, in the distinct elements problem, the adversary's goal is that at some step $i$, the estimate $z_i$ will fail to be a $(1+\alpha)$-approximation of the true current number of distinct elements.

\subsection{Our results}
Loosely speaking, we show a reduction from the problem of adaptive data analysis (ADA) to the problem of adversarial streaming. Our results then follow from known impossibility results in the adaptive data analysis literature. In the ADA problem, given a sample $S$ containing $n$ independent samples from some unknown distribution $\DDD$ over a domain $X$, the goal is to provide answers to a sequence of adaptively chosen queries w.r.t.\ $\DDD$. 
Importantly, the answers must be accurate w.r.t.\ the (unknown) underlying distribution $\DDD$; not just w.r.t.\ the empirical sample $S$. In more detail, in the ADA problem, on every time step $i$ we get a query $q_i:X\rightarrow\{0,1\}$, and we need to respond with an answer $a_i$ that approximates $q_i(\DDD)\triangleq\E_{x\sim\DDD}[q_i(x)]$. Observe that if all of the queries were fixed before the sample $S$ is drawn, then we could simply answer each query $q_i$ with its empirical average $q_i(S)\triangleq\frac{1}{n}\sum_{x\in S}q_i(x)$. Indeed, by the Hoeffding bound, in such a case these answers provide good approximations to the true answers $q_i(\DDD)$. Furthermore, the number of queries $\ell$ that we can support can be exponential in the sample size $n$.
However, this argument breaks completely when the queries are chosen adaptively based on previous answers given by the mechanism, and the problem becomes much more complex. While, information-theoretically, it is still possible to answer an exponential number of queries, it is known that every {\em computationally efficient} mechanism cannot answer more than $n^2$ adaptive queries using a sample of size $n$.

We show that the ADA problem can be phrased as a streaming problem, where the first $n$ elements in the stream are interpreted as ``data points'' and later elements in the stream are interpreted as ``queries''. In order to apply existing impossibility results for the ADA problem, we must overcome the following two main challenges.

\paragraph{Challenge 1 and its resolution.} The difficulty in the ADA problem is to maintain accuracy w.r.t.\ the unknown underlying distribution (and not just w.r.t.\ the input sample, which is easy). In the streaming setting, however, there is no underlying distribution, and we cannot require a streaming algorithm to be accurate w.r.t.\ such a distribution.
Instead, we require the streaming algorithm to give accurate answers only w.r.t.\ the input sample (i.e., w.r.t.\ the dataset defined by the first $n$ elements in the stream). This makes our streaming problem well-defined. We then show that if these $n$ elements are sampled i.i.d.\ from some underlying distribution, then we can use {\em compression arguments} to show that if the streaming algorithm has small space complexity, and if its answers are accurate w.r.t.\ the empirical sample, then its answers must in fact be accurate also w.r.t.\ this underlying distribution. In other words, even though we only require the streaming algorithm to give accurate answers w.r.t.\ the empirical sample, we show that if it uses small space complexity then its answers must {\em generalize} to the underlying distribution. This allows us to formulate a link to the ADA problem. We remark that, in the actual construction, we need to introduce several technical modifications in order to make sure that the resulting streaming problem can be solved with small space complexity in the oblivious setting.

\paragraph{Challenge 2 and its resolution.} The impossibility results we mentioned for the ADA problem only hold for {\em computationally efficient} mechanisms, where we aim for an information-theoretic separation. While there exist information-theoretic impossibility results for the ADA problem, they are too weak to give a meaningful result in our context. We therefore cannot apply existing negative results for the ADA problem to our setting as is. Informally, the reason that the negative results for the ADA problem only hold for computationally efficient mechanisms is that their constructions rely on the existence of an efficient {\em encryption scheme} whose security holds under computational assumptions. We observe that, in our setting, we can replace this encryption scheme with a different scheme with information-theoretic security against adversaries with {\em bounded storage} capabilities. Indeed, in our setting, the ``adversary'' for this encryption scheme will be the streaming algorithm, whose storage capabilities are bounded.

\medskip
We obtain the following theorem.

\begin{theorem}\label{thm:mainIntro}
For every $w$, there exists a streaming problem over domain of size $\poly(w)$ and stream length $O(w^5)$ that requires at least $w$ space to be solved in the adversarial setting to within constant accuracy (small enough), but can be solved in the oblivious setting using space $O(\log^2(w))$.
\end{theorem}

\subsubsection{Optimality of our results in terms of the flip-number} 
The previous works of \cite{BenEliezerJWY20,HassidimKMMS20,WoodruffZhou20} stated their positive results in terms of the following definition.

\begin{definition}[Flip number \cite{BenEliezerJWY20}]
Let $g$ be a function defining a streaming problem. 
The {\em $(\alpha,m)$-flip number} of $g$, denoted as $\lambda$, is the maximal number of times that the value of the value of $g$ can change (increase or decrease) by a factor of $(1+\alpha)$ during a stream of length $m$.
\end{definition}

The works of~\cite{BenEliezerJWY20,WoodruffZhou20,HassidimKMMS20} presented general frameworks form transforming an oblivious streaming algorithm $\AAA$ into an adversarially robust streaming algorithm $\BBB$ with space complexity (roughly) $\sqrt{\lambda}\cdot{\rm Space}(\AAA)$. That is, the results of~\cite{BenEliezerJWY20,WoodruffZhou20,HassidimKMMS20} showed that, generally, adversarial robustness requires space blowup at most (roughly) $\sqrt{\lambda}$ compared to the oblivious setting. For the streaming problem we present (see Theorem~\ref{thm:mainIntro}) it holds that the flip-number is $O(w^2)$. That is, for every $w$, we present a streaming problem with flip-number $\lambda=O(w^2)$, that requires at least $w=\Omega(\sqrt{\lambda})$ space to be solved in the adversarial setting to within constant accuracy (small enough), but can be solved in the oblivious setting using space $O(\log^2(w))$. This means that, in terms of the dependency of the space complexity in the flip-number, our results are nearly tight. In particular, in terms of $\lambda$, our results show that a blowup of $\tilde{\Omega}(\sqrt{\lambda})$ to the space complexity is generally unavoidable in the adversarial setting.

\subsubsection{A reduction from adaptive data analysis to adversarial streaming}
Informally, we consider the following streaming problem, which we call the Streaming Adaptive Data Analysis (SADA) problem. On every time step $i\in[m]$ we get an update $x_i \in X$. 
We interpret the first $n$ updates in the stream $x_1,\dots,x_n$ as ``data points'', defining a multiset $S=\{x_1,\dots,x_n\}$. This multiset does not change after time $n$.

The next updates in the stream (starting from time $i=n+1$) define ``queries'' $q:X\rightarrow\{0,1\}$ that should be evaluated by the streaming algorithm on the multiset $S$. That is, for every such query $q$, the streaming algorithm should respond with an approximation of $q(S)=\frac{1}{n}\sum_{x\in S}q(x)$. A technical issue here is that every such query is represented using $|X|$ bits, and hence, cannot be specified using a single update in the stream (which only consists of $\log|X|$ bits). Therefore, every query is specified using $|X|$ updates in the stream. Specifically, starting from time $i=n+1$, every bulk of $|X|$ updates defines a query $q:X\rightarrow\{0,1\}$. At the end of every such bulk, the goal of the streaming algorithm is to output (an approximation for) the average of $q$ on the multiset $S$. On other time steps, the streaming algorithm should output 0.

As we mentioned, we use {\em compression arguments} to show that if the streaming algorithm is capable of accurately approximating the average of every such query on the multiset $S$, and if it uses small space, then when the ``data points'' (i.e., the elements in the first $n$ updates) are sampled i.i.d.\ from some distribution $\DDD$ on $X$, then the answers given by the streaming algorithm must in fact be accurate also w.r.t.\ the expectation of these queries on $\DDD$. This means that the existence of streaming algorithm for the SADA problem implies the existence of an algorithm for the adaptive data analysis (ADA) problem, with related parameters. Applying known impossibility results for the ADA problem, this results in a contradiction. However, as we mentioned, the impossibility results we need for the ADA problem only hold for {\em computationally efficient} mechanisms. Therefore, the construction outlined here only rules out {\em computationally efficient} adversarially-robust streaming algorithms for the SADA problem. To get an information-theoretic separation, we modify the definition of the SADA problem and rely on cryptographic techniques from the {\em bounded storage} model. 

\begin{remark}
In Section~\ref{sec:SADA2} we outline a variant of the SADA problem, which is more ``natural'' in some sense, but for which we can only show a computational separation.
\end{remark}

\section{Preliminaries}

Our results rely on tools and techniques from the literature on adaptive data analysis, computational learning (in particular {\em compression}), and cryptography (in particular {\em pseudorandom generators} and {\em encryption schemes}). We now introduce the needed preliminaries.

\subsection{Adaptive data analysis}

A \emph{statistical query} over a domain $X$ is specified by a predicate $q:X\rightarrow\{0,1\}$. The value of a query $q$ on a distribution $\DDD$ over $X$ is $q(\DDD)=\E_{x\sim\DDD}[q(x)]$.
Given a database $S\in X^n$ and a query $q$, we denote the empirical average of $q$ on $S$ as $q(S)=\frac{1}{n}\sum_{x\in S}{q(x)}$.

In the adaptive data analysis (ADA) problem, the goal is to design a \emph{mechanism} $\cM$ that answers queries on $\DDD$ using only i.i.d.\ samples from it. Our focus is the case where the queries are chosen adaptively and adversarially.
Specifically, $\cM$ is a stateful algorithm that holds a collection of samples $(x_1,\dots,x_n)$, takes a statistical query $q$ as input, and returns an answer $z$. We require that when $x_1,\dots,x_n$ are independent samples from $\DDD$, then the answer $z$ is close to $q(\DDD)$. Moreover we require that this condition holds for every query in an adaptively chosen sequence $q_1,\dots, q_{\ell}$. Formally, we define an accuracy game $\accgame_{n, \ell,\cM,\mathbb{A}}$ between a mechanism $\cM$ and a stateful \emph{adversary} $\mathbb{A}$ in Algorithm~\ref{fig:accgame1}.

\begin{definition}\label{def:accuratemechanism}
A mechanism $\cM$ is \emph{$(\alpha,\beta)$-statistically-accurate for $\ell$ adaptively chosen statistical queries given $n$ samples} if for every adversary $\mathbb{A}$ and every distribution $\DDD$,
\begin{equation}
\Pr_{\substack{S\sim\DDD^n\\ \accgame_{n, \ell,\cM, \mathbb{A}}(S)}}\left[
\max_{ i \in [\ell]}\left| q_i(\DDD)-z_i \right|\leq\alpha
\right]\geq1-\beta.
\label{eq:accuratemechanism}
\end{equation}
\end{definition}

\begin{remark}
Without loss of generality, in order to show that a mechanism $\cM$ is $(\alpha,\beta)$-statistically-accurate (as per Definition~\ref{def:accuratemechanism}), it suffices to consider only {\em deterministic} adversaries $\mathbb{A}$. Indeed, given a randomized adversary $\mathbb{A}$, if requirement~(\ref{eq:accuratemechanism}) holds for every fixture of its random coins, then it also holds when the coins are random.
\end{remark}

We use a similar definition for empirical accuracy:

\begin{definition}\label{def:EmpAccurate}
A mechanism $\cM$ is \emph{$(\alpha,\beta)$-empirically accurate for $\ell$ adaptively chosen statistical queries given a database of size $n$} if for every adversary $\mathbb{A}$ and every database $S$ of size $n$,
\begin{equation*}
\Pr_{\substack{ \accgame_{n, \ell,\cM, \mathbb{A}}(S)}}\left[
\max_{ i \in [\ell]}\left| q_i(S)-z_i \right|\leq\alpha
\right]\geq1-\beta.
\end{equation*}
\end{definition}

\begin{algorithm*}[t!]
\caption{The Accuracy Game $\accgame_{n, \ell,\cM,\mathbb{A}}$.}\label{fig:accgame1}

{\bf Input:} A database $S\in X^n$.

\begin{enumerate}[leftmargin=15pt,rightmargin=10pt,itemsep=1pt,topsep=3pt]

\item The database $S$ is given to $\cM$.

\item For $i=1$ to $\ell$, 

\begin{enumerate}
	\item The adversary $\mathbb{A}$ chooses a statistical query $q_i$.
	\item The mechanism $\cM$ gets $q_i$ and outputs an answer $z_i$.
	\item The adversary $\mathbb{A}$ gets $z_i$.
\end{enumerate}

\item Output the transcript $(q_1,z_1,\dots,q_{\ell},z_{\ell})$.

\end{enumerate}
\end{algorithm*}

\subsection{Transcript compressibility}

An important notion that allows us to argue about the utility guarantees of an algorithm that answers adaptively chosen queries is {\em transcript compressibility}, defined as follows.

\begin{definition}[\cite{DworkFHPRR-nips-2015}]
A mechanism $\cM$ is {\em transcript compressible} to $b(n,\ell)$ bits if for every deterministic adversary $\mathbb{A}$ there is a set of transcripts $H_{\mathbb{A}}$ of size $\left|H_{\mathbb{A}}\right|\leq2^{b(n,\ell)}$ such that for every dataset $S\in X^n$ we have
$$
\Pr\left[ \accgame_{n, \ell,\cM, \mathbb{A}}(S)\in H_{\mathbb{A}} \right]=1.
$$
\end{definition}

The following theorem shows that, with high probability, for every query generated throughout the interaction with a transcript compressible mechanism it holds that its empirical average is close to its expectation.

\begin{theorem}[\cite{DworkFHPRR-nips-2015}]\label{thm:transcriptCompression}
Let $\cM$ be transcript compressible to $b(n,\ell)$ bits, and let $\beta>0$. Then, for every adversary $\mathbb{A}$ and for every distribution $\DDD$ it holds that
$$
\Pr_{\substack{S\sim\DDD^n\\ \accgame_{n, \ell,\cM, \mathbb{A}}(S)}}\left[ \exists i \text{ such that } \left|q_i(S)-q_i(\DDD)\right|>\alpha  \right]\leq\beta,
$$
where
$$
\alpha=O\left( \sqrt{\frac{b(n,\ell)+\ln(\ell/\beta)}{n}}  \right).
$$
\end{theorem}

\subsection{Pseudorandom generators in the bounded storage model}\label{sec:BSM}

Our results rely on the existence of pseudorandom generators providing information-theoretic security against adversaries with bounded storage capabilities. This security requirement is called the {\em bounded storage model}. This model was introduced by Maurer~\cite{maurer1992conditionally}, and has generated many interesting results, e.g.,~\cite{maurer1992conditionally,cachin1997unconditional,aumann1999information,aumann2002everlasting,ding2002hyper,dziembowski2004optimal,lu2004encryption,harnik2006everlasting}. We give here the formulation presented by Vadhan~\cite{vadhan2004}.

The bounded storage model utilizes a short seed $K\in\{0,1\}^b$ (unknown to the adversary) and a long stream of public random bits $X_1,X_2,\dots$ (known to all parties). A {\em bounded storage model (BSM) pseudorandom generator} is a function
$\PRG:\{0,1\}^a\times\{0,1\}^b\rightarrow\{0,1\}^c$, typically with $b,c\ll a$. Such a scheme is to be used as follows. 
Initially, two (honest) parties share a seed $K\in\{0,1\}^b$ (unknown to the adversary). At time $t\in[T]$, the next $a$ bits of the public stream $(X_{(t-1)a},\dots,X_{ta})$ are broadcast. The adversary is allowed to listen to this stream, however, it cannot store all of it as it has bounded storage capabilities. The honest parties apply $\PRG(\cdot,K)$ to this stream obtain $c$ pseudorandom bits, denoted as $Y_t\in\{0,1\}^c$.

We now formally define security for a BSM pseudorandom generator. Let $\beta a$ be the bound on the storage of the adversary $\AAA$, and denote by $S_t\in\{0,1\}^{\beta a}$ the state of the adversary at time $t$. We consider the adversary's ability to distinguish two experiments --- the ``real'' one, in which the pseudorandom generator is used, and an ``ideal'' one, in which truly random bits are used. Let $\AAA$ be an arbitrary function representing the way the adversary updates its storage and attempts to distinguish the two experiments at the end.

\paragraph{Real Experiment:}
\begin{itemize}
	\item Let $X=(X_1,X_2,\dots,X_{Ta})$ be a sequence of uniformly random bits, let $K\leftarrow\{0,1\}^b$ be the key, and let the adversary's initial state by $S_0=0^{\beta a}$. 
	\item For $t=1,\dots,T$:\\ -- Let $Y_t=\PRG\left( X_{(t-1)a+1},\dots,X_{ta}, K \right)\in\{0,1\}^c$ be the pseudorandom bits.\\ -- Let $S_t=\AAA\left( Y_1,\dots,Y_{t-1}, S_{t-1},  X_{(t-1)a+1},\dots,X_{ta} \right)\in\{0,1\}^{\beta a}$ be the adversary's new state.
	\item Output $\AAA\left( Y_1,\dots,Y_T,  S_T,  K  \right)$
\end{itemize}

\paragraph{Ideal Experiment:}
\begin{itemize}
	\item Let $X=(X_1,X_2,\dots,X_{Ta})$ be a sequence of uniformly random bits, let $K\leftarrow\{0,1\}^b$ be the key, and let the adversary's initial state by $S_0=0^{\beta a}$. 
	\item For $t=1,\dots,T$:\\ -- Let $Y_t\leftarrow\{0,1\}^c$ be truly random bits.\\ -- Let $S_t=\AAA\left( Y_1,\dots,Y_{t-1}, S_{t-1},  X_{(t-1)a+1},\dots,X_{ta} \right)\in\{0,1\}^{\beta a}$ be the adversary's new state.
	\item Output $\AAA\left( Y_1,\dots,Y_T,  S_T,  K  \right)\in\{0,1\}$.
\end{itemize}

Note that at each time step we give the adversary access to all the past $Y_i$'s ``for free'' (i.e. with no cost in the storage bound), and in the last time step, we give the adversary the adversary the key $K$.

\begin{definition}[\cite{vadhan2004}]
We call $\PRG:\{0,1\}^a\times\{0,1\}^b\rightarrow\{0,1\}^c$ an {\em $\varepsilon$-secure BSM pseudorandom generator} for storage rate $\beta$ if for every adversary $\AAA$ with storage bound $\beta a$, and every $T\in\N$, the adversary $\AAA$ distinguishes between the
real and ideal experiments with advantage at most $T\varepsilon$. That is, 
$$
\left|
\Pr_{\rm real}\left[ \AAA\left( Y_1,\dots,Y_T,  S_T,  K  \right)=1  \right]
-
\Pr_{\rm ideal}\left[ \AAA\left( Y_1,\dots,Y_T,  S_T,  K  \right)=1  \right]
\right|\leq T\cdot\varepsilon
$$
\end{definition}

\begin{remark}
No constraint is put on the computational power of the adversary except for the storage bound of $\beta a$ (as captured by $S_t\in\{0,1\}^{\beta a}$). This means that the distributions of
$(Y_1,\dots,Y_T,  S_T,  K)$ in the real and ideal experiments are actually close in a statistical sense -- they must have statistical difference at most $T\cdot\varepsilon$.
\end{remark}

We will use the following result of Vadhan~\cite{vadhan2004}. We remark that this is only a special case of the results of Vadhan, and refer the reader to~\cite{vadhan2004} for a more detailed account.

\begin{theorem}[\cite{vadhan2004}]\label{thm:Vadhan}
For every $a\in\N$, every $\varepsilon>\exp\left( -a/2^{O(\log^{*}a)} \right)$, and every $c\leq a/4$, there is a BSM pseudorandom generator $\PRG:\{0,1\}^a \times\{0,1\}^b\rightarrow\{0,1\}^c$ such that
\begin{enumerate}
	\item $\PRG$ is $\varepsilon$-secure for storage rate $\beta\leq1/2$.
	\item $\PRG$ has key length $b=O(\log(a/\varepsilon))$.
	\item For every key $K$, $\PRG(\cdot,K)$ reads at most $t=O(c+\log(1/\varepsilon))$ bits from the public stream (nonadaptively).
	\item $\PRG$ is computable in time $\poly(t,b)$ and uses workspace $\poly(\log t, \log b)$ in addition to the
$t$ bits read from the public stream and the key of length $b$.
\end{enumerate}
\end{theorem}

\section{The Streaming Adaptive Data Analysis (SADA) Problem}\label{sec:SADA}

In this section we introduce a streaming problem, which we call the Streaming Adaptive Data Analysis (SADA) problem, for which we show a strong positive result in the oblivious setting and a strong negative result in the adversarial setting.

Let $X=\{0,1\}^d\times\{0,1\}^b$ be a data domain, let $\gamma\geq 0$ be a fixed constant,
and let $\PRG:\{0,1\}^a\times\{0,1\}^b\rightarrow\{0,1\}^c$ be a BSM pseudorandom generator, where $c=1$. We consider the following streaming problem. On every time step $i\in[m]$ we get an update $x_i=(p_i,k_i)\in X$. 
We interpret the first $n$ updates in the stream $x_1,\dots,x_{n}$ as pairs of ``data points'' and their corresponding ``keys''. Formally, we denote by $S$ the multiset containing the pairs $x_1,\dots,x_{n}$. For technical reasons, the multiset $S$ also contains $\frac{\gamma n}{1-\gamma}$ copies of some arbitrary element $\bot$. This multiset does not change after time $n$.

Starting from time $j=n+1$, each bulk of $(a+1)\cdot 2^{d}$ updates re-defines a ``function'' (or a ``query'') that should be evaluated by the streaming algorithm on the multiset $S$. This function is defined as follows.

\begin{enumerate}
	\item For $p\in\{0,1\}^d$ (in lexicographic order) do
	\begin{enumerate}
		\item Let $x^1,\dots,x^{a}\in X$ denote the next $a$ updates, and let $\Gamma\in\{0,1\}^a$ be the bitstring containing the first bit of every such update.
		\item Let $x$ denote the next update, and let $\sigma$ denote its first bit.
		\item For every $k\in\{0,1\}^b$, let $Y_k=\PRG(\Gamma,k)$ and define $f(p,k)=\sigma\oplus Y_k$.
	\end{enumerate}
	\item Also set $f(\bot)=1$.
\end{enumerate}

This defines a function $f:\{0,1\}^d \times \{0,1\}^b\rightarrow\{0,1\}$. 

\begin{definition}[The $(a,b,d,m,n,\gamma)$-SADA Problem]
At the end of every such bulk, defining a function $f$, the goal of the streaming algorithm is to output (an approximation for) the average of $f$ on the multiset $S$. On other time steps, the streaming algorithm should output 0.
\end{definition}

\begin{remark}
In the definition above, $m$ is the total number of updates (i.e., the length of the stream), $n$ the number of updates that we consider as ``date points'', $\gamma$ is a small constant, and $a,b,d$ are the parameters defining the domain (and the $\PRG$).
\end{remark}

\section{An Oblivious Algorithm for the SADA Problem}\label{sec:SADAObl}

In the oblivious setting, we can easily construct a streaming algorithm for the SADA problem using {\em sampling}. Specifically, throughout the first phase of the execution (during the first $n$ time steps) we maintain a small representative sample from the ``data items'' (and their corresponding ``keys'') from the stream. In the second phase of the execution we use this sample in order to answer the given queries. Consider Algorithm \texttt{ObliviousSADA}, specified in Algorithm~\ref{alg:ObliviousSADA}. We now analyze its utility guarantees.

\begin{algorithm*}[t!]
\caption{\bf \texttt{ObliviousSADA}}\label{alg:ObliviousSADA}

{\bf Setting:} On every time step we obtain the next update, which is an element of $X=\{0,1\}^d\times\{0,1\}^b$.

{\bf Algorithm used:} A sampling algorithm $\SAMP$ that operates on a stream of elements from the domain $X$ and maintains a representative sample.

\begin{enumerate}[leftmargin=15pt,rightmargin=10pt,itemsep=1pt,topsep=3pt]

\item Instantiate algorithm $\SAMP$.

\item REPEAT $n$ times
\begin{enumerate}
	\item Obtain the next update in the stream $x=(p,k)$.
	\item Output $0$.
	\item Feed the update $x$ to $\SAMP$.
\end{enumerate}

\item Feed (one by one) $\frac{\gamma n}{1-\gamma}$ copies of $\bot$ to $\SAMP$.

\item Let $D$ denote the sample produced by algorithm $\SAMP$.

\item REPEAT (each iteration of this loop spans over $2^d (a+1)$ updates that define a query) %

\begin{enumerate}
	\item Let $v$ denote the multiplicity of $\bot$ in $D$, and set $F=\frac{v}{|D|}$.
	
	\item For every $p\in\{0,1\}^d$ in lexicographic order do (inner loop) %
	\begin{enumerate}
		\item Denote $K_p = \{ k : (p,k)\in D \}$. That is, $K_p$ is the set of all keys $k$ such that $(p,k)$ appears in the sample $D$.
		\item REPEAT $a$ times
		\begin{itemize}
			\item Obtain the next update $x$
			\item For every $k\in K_p$, feed the first bit of $x$ to $\PRG(\cdot,k)$. %
			\item Output 0.
		\end{itemize}
		\item For every $k\in K_p$, obtain a bit $Y_k$ from $\PRG(\cdot,k)$.

		\item Obtain the next update and let $\sigma$ be its first bit (and output 0).
		
		\item For every $k\in K_p$ such that $\sigma\oplus Y_k=1$: Let $v_{(p,k)}$ denote the multiplicity of $(p,k)$ in $D$, and set $F\leftarrow F + \frac{v_{(p,k)}}{|D|}$.
	\end{enumerate}
	
	\item Output $F$.
	
\end{enumerate}

\end{enumerate}
\end{algorithm*}

\begin{theorem}
Algorithm \texttt{ObliviousSADA} is $(\alpha,\beta)$-accurate for the SADA problem in the oblivious setting.
\end{theorem}

\begin{proof}
Fix the stream $\vec{x}_m=(x_1,\dots,x_m)$.
We will assume that \texttt{ObliviousSADA} is executed with a sampling algorithm $\SAMP$ that returns a sample $D$ containing $|D|$ elements, sampled uniformly and independently from $S=(x_1,\dots,x_n,\bot,\dots,\bot)$. This can be achieved, e.g., using Reservoir Sampling~\cite{Vitter85}. As the stream is fixed (and it is of length $m$), there are at most $m$ different queries that are specified throughout the execution. By the Chernoff bound, assuming that $|D|\geq\Omega\left( \frac{1}{\alpha^2\gamma}\ln(\frac{m}{\beta}) \right)$, with probability at least $1-\beta$, for every query $f$ throughout the execution we have that $f(D)\in(1\pm\alpha)\cdot f(S)$. The theorem now follows by observing that the answers given by algorithm \texttt{ObliviousSADA} are exactly the empirical average of the corresponding queries on $D$.
\end{proof}

\begin{observation}\label{obs:obliviousSpace}
For constant $\alpha,\beta,\gamma$, using the pseudorandom generator from Theorem~\ref{thm:Vadhan}, algorithm \texttt{ObliviousSADA} uses space $O\left( \left( \log(\frac{1}{\eps}) + b+d\right)\cdot\log(m) \right)$.
\end{observation}

\begin{proof}
The algorithm maintains a sample $D$ containing $O(\log m)$ elements, where each element is represented using $b+d$ bits. In addition, the pseudorandom generator uses $O( \log(\frac{1}{\eps}) )$ bits of memory, and the algorithm instantiates at most $|D|=O(\log m)$ copies of it.
\end{proof}

\section{An Impossibility Result for Adaptive Streaming}\label{sec:negative}

Suppose that there is an adversarially robust streaming algorithm $\AAA$ for the SADA problem. We use $\AAA$ to construct an algorithm that gets a sample $P$ containing $n$ points in $\{0,1\}^d$, and answers adaptively chosen queries $q:\{0,1\}^d\rightarrow\{0,1\}$. Consider Algorithm \texttt{AnswerQueries}, specified in Algorithm~\ref{alg:AnswerQueries}.

\begin{algorithm*}[t!]
\caption{\bf \texttt{AnswerQueries}}\label{alg:AnswerQueries}

{\bf Input:} A database $P\in(\{0,1\}^d)^n$ containing $n$ elements from $\{0,1\}^d$.

{\bf Setting:} On every time step we get a query $q:\{0,1\}^d\rightarrow\{0,1\}$.

{\bf Algorithm used:} An adversarially robust streaming algorithm $\AAA$ for the SADA problem with $(\alpha,\beta)$-accuracy for streams of length $m$. We abstract the coin tosses of $\AAA$ using {\em two} random strings, $r_1$ and $r_2$, of possibly unbounded length. Initially, we execute $\AAA$ with access to $r_1$, meaning that every time it tosses a coin it gets the next bit in $r_1$. At some point, we switch the random string to $r_2$, and henceforth $\AAA$ gets its coin tosses from $r_2$.

{\bf Algorithm used:} BSM pseudorandom generator $\PRG:\{0,1\}^a\times\{0,1\}^b\rightarrow\{0,1\}$, as in the definition of the SADA problem.

\begin{enumerate}[leftmargin=15pt,rightmargin=10pt,itemsep=1pt,topsep=3pt]

\item For every $p\in\{0,1\}^d$ sample $k_p\in\{0,1\}^b$ uniformly.

\item Sample $r_1\in\{0,1\}^{\nu}$ uniformly, and instantiate algorithm $\AAA$ with read-once access to bits of $r_1$. Here $\nu$ bounds the number of coin flips made by $\AAA$.

\item For every $p\in P$, feed the update $(p,k_p)$ to $\AAA$.

\item Sample $r_2\in\{0,1\}^{\nu}$ uniformly, and switch the read-once access of $\AAA$ to $r_2$. (The switch from $r_1$ to $r_2$ is done for convenience, so that after Step~3 we do not need to ``remember'' the position for the next coin from $r_1$.)

\item REPEAT $\ell \triangleq\frac{m-n}{(a+1)\cdot 2^d}$ times 

\begin{enumerate}
	\item Obtain the next query $q:\{0,1\}^d\rightarrow\{0,1\}$.
	
	\item For every $p\in\{0,1\}^d$ do
	
	\begin{enumerate}
		\item Sample $\Gamma\in\{0,1\}^a$ uniformly.
		\item Feed $a$ updates (one by one) to $\AAA$ s.t.\ the concatenation of their first bits is $\Gamma$.
		\item Let $Y=\PRG(\Gamma,k_p)$.
		\item Feed to $\AAA$ an update whose first bit is $Y\oplus q(p)$.
	\end{enumerate}
	
	\item Obtain an answer $z$ from $\AAA$.
	\item Output $z$.
	
\end{enumerate}

\end{enumerate}
\end{algorithm*}

By construction, assuming that $\AAA$ is accurate for the SADA problem, we get that \texttt{AnswerQueries} is empirically-accurate (w.r.t.\ its input database $P$). Formally, 

\begin{claim}\label{claim:empiricalAccuracy}
If $\AAA$ is $(\alpha,\beta)$-accurate for the SADA problem, then $\texttt{AnswerQueries}$ is $\left(\frac{\alpha}{1-\gamma},\beta\right)$-empirically-accurate for $\frac{m-n}{(a+1)\cdot2^d}$ adaptively chosen statistical queries given a database of size $n$.
\end{claim}

\begin{proof}[Proof sketch]
Let $q$ denote the query given at some iteration, and let $f$ denote the corresponding function specified to algorithm $\AAA$ during this iteration. The claim follows from the fact that, by construction, for every $(p,k)$ we have that $f(p,k)=q(p)$. 
Specifically, w.h.p., the answers given by $\AAA$ are $\alpha$-accurate w.r.t.\ $P\cup\{\bot,\dots,\bot\}$, and hence, $\frac{\alpha}{1-\gamma}$-accurate w.r.t.\ $P$.
\end{proof}

We now show that algorithm \texttt{AnswerQueries} is transcript-compressible. To that end, for every choice of $\vec{\Gamma},\vec{k},r_1,r_2$ for the strings $\Gamma$, the keys $k$, and the random bitstrings $r_1,r_2$ used throughout the execution, let us denote by $\texttt{AnswerQueries}_{\vec{\Gamma},\vec{k},r_1,r_2}$ algorithm \texttt{AnswerQueries} after fixing these elements.

\begin{claim}\label{claim:compress}
If algorithm $\AAA$ uses space at most $w$, then, for every $\vec{\Gamma},\vec{k},r_1,r_2$, we have that algorithm $\texttt{AnswerQueries}_{\vec{\Gamma},\vec{k},r_1,r_2}$ is transcript-compressible to $w$ bits.
\end{claim}

\begin{proof}[Proof sketch]
Assuming that the adversary who generates the queries $q$ is deterministic (which is without loss of generality) we get that the entire transcript is determined by the state of algorithm $\AAA$ at the end of Step 3.
\end{proof}

\begin{remark}
The ``switch'' from $r_1$ to $r_2$ is convenient in the proof of Claim~\ref{claim:compress}. Otherwise, in order to describe the state of the algorithm after Step~3 we need to specify both the internal state of $\AAA$ and the position for the next coin from $r_1$.
\end{remark}

Combining Claims~\ref{claim:empiricalAccuracy} (empirical accuracy), and~\ref{claim:compress} (transcript-compression), we get the following lemma.

\begin{lemma}\label{lem:streamingAccuracy}
Suppose that $\AAA$ is $(\alpha,\beta)$-accurate for the SADA problem for streams of length $m$ using memory $w$. Then for every $\beta'>0$, algorithm \texttt{AnswerQueries} is $\left(\frac{\alpha}{1-\gamma}+\alpha',\beta+\beta'\right)$-statistically-accurate for $\ell=\frac{m-n}{(a+1)\cdot2^d}$ queries, where
$$
\alpha'=O\left( \sqrt{\frac{w+\ln(\frac{\ell}{\beta'})}{n}}\right).
$$
\end{lemma}

\begin{proof}
Fix a distribution $\DDD$ over $\{0,1\}^d$ and fix an adversary $\mathbb{A}$ that generates the queries $q_i$. Consider the execution of the accuracy game $\accgame$ (given in Algorithm~\ref{fig:accgame1}). By Claim~\ref{claim:empiricalAccuracy}, 
$$
\Pr_{\substack{S\sim\DDD^n\\ \accgame_{n, \ell,\texttt{AnswerQueries}, \mathbb{A}}(S)}}\left[ \exists i \text{ such that } \left|q_i(S)-z_i\right|>\frac{\alpha}{1-\gamma}  \right]\leq\beta,
$$
where the $z_i$'s denote the answers given by the algorithm. In addition, by Claim~\ref{claim:compress} and Theorem~\ref{thm:transcriptCompression}, for every fixing of $\vec{\Gamma},\vec{k},r_1,r_2$ we have that
\begin{equation}
\Pr_{\substack{S\sim\DDD^n\\ \accgame_{n, \ell,\texttt{AnswerQueries}, \mathbb{A}}(S)}}\Big[ \exists i \text{ such that } \left|q_i(S)-q_i(\DDD)\right|>\alpha'  \Big|  \vec{\Gamma},\vec{k},r_1,r_2  \Big]\leq\beta',
\label{eq:transcriptCompression}
\end{equation}
where
$$
\alpha'=O\left( \sqrt{\frac{w+\ln(\ell/\beta')}{n}}  \right).
$$
Since Inequality~(\ref{eq:transcriptCompression}) holds for {\em every} fixing of $\vec{\Gamma},\vec{k},r_1,r_2$, it also holds when sampling them. Therefore, by the triangle inequality and the union bound,
\begin{align*}
&\Pr_{\substack{S\sim\DDD^n\\ \accgame_{n, \ell,\texttt{AnswerQueries}, \mathbb{A}}(S)}}\left[ \exists i \text{ such that } \left|z_i-q_i(\DDD)\right|>\frac{\alpha}{1-\gamma} + \alpha'  \right]\\
&\quad \leq \Pr_{\substack{S\sim\DDD^n\\ \accgame_{n, \ell,\texttt{AnswerQueries}, \mathbb{A}}(S)}}\left[ \exists i \text{ such that } \left|q_i(S)-z_i\right|>\frac{\alpha}{1-\gamma} \text{ or } \left|q_i(S)-q_i(\DDD)\right|>\alpha' \right]\leq\beta+\beta'.
\end{align*}
\end{proof}

To obtain a contradiction, we rely on the following impossibility result for the ADA problem. 
Consider an algorithm $\cM$ for the ADA problem that gets an input sample $P=(p_1,\dots,p_n)$ and answers (adaptively chosen) queries $q$.
The impossibility result we use states that if $\cM$ computes the answer to every given query $q$ only as a function of the value of $q$ on points from $P$ (i.e., only as a function of $q(p_1),\dots,q(p_n)$), then, in general, $\cM$ cannot answer more than $n^2$ adaptively chosen queries. An algorithm $\cM$ satisfying this restriction is called a {\em natural} mechanism. Formally,

\begin{definition}[\cite{HardtU14}]
An algorithm that takes a sample $P$ and answers queries $q$ is {\em natural} if for every input sample $P$ and every two queries $q$ and $q'$ such that $q(p) = q'(p)$ for all $p\in P$, the answers $z$ and $z'$ that the algorithm gives on queries $q$ and $q'$, respectively, are identical if the algorithm is deterministic and identically distributed if the algorithm is randomized. If the algorithm is stateful, then this condition should hold when the algorithm is in any of its possible states.
\end{definition}

We will use the following negative result of Steinke and Ullman~\cite{SU15} (see also~\cite{HardtU14, NissimSSSU18}).

\begin{theorem}[\cite{SU15}]\label{thm:adaNegative}
There exists a constant $c>0$ such that there is no natural algorithm that is $(c,c)$-statistically-accurate for $O(n^2)$ adaptively chosen
queries given $n$ samples over a domain of size $\Omega(n)$.
\end{theorem}

We have already established (in Lemma~\ref{lem:streamingAccuracy}) that algorithm \texttt{AnswerQueries} is statistically-accurate for $\ell=\frac{m-n}{(a+1)\cdot2^d}$ adaptively chosen queries, where $\ell$ can easily be made bigger then $n^2$ (by taking $m$ to be big enough). We now want to apply Theorem~\ref{thm:adaNegative} to our setting in order to get a contradiction. However, algorithm \texttt{AnswerQueries} is not exactly a natural algorithm (though, as we next explain, it is very close to being natural). The issue is that the answers produced by the streaming algorithm $\AAA$ can (supposedly) depend on the value of the given queries outside of the input sample. Therefore, we now tweak algorithm \texttt{AnswerQueries} such that it becomes a natural algorithm. The modified construction is given in Algorithm \texttt{AnswerQueriesOTP}, where we marked the modifications in red. Consider Algorithm \texttt{AnswerQueriesOTP}, specified in Algorithm~\ref{alg:AnswerQueriesOTP}.

\begin{algorithm*}[t!]
\caption{\bf \texttt{AnswerQueriesOTP}}\label{alg:AnswerQueriesOTP}

{\bf Input:} A database $P\in(\{0,1\}^d)^n$ containing $n$ elements from $\{0,1\}^d$.

{\bf Setting:} On every time step we get a query $q:\{0,1\}^d\rightarrow\{0,1\}$.

{\bf Algorithm used:} An adversarially robust streaming algorithm $\AAA$ for the SADA problem with $(\alpha,\beta)$-accuracy for streams of length $m$. We abstract the coin tosses of $\AAA$ using {\em two} random strings, $r_1$ and $r_2$, of possibly unbounded length. Initially, we execute $\AAA$ with access to $r_1$, meaning that every time it tosses a coin it gets the next bit in $r_1$. At some point, we switch the random string to $r_2$, and henceforth $\AAA$ gets its coin tosses from $r_2$.

{\bf Algorithm used:} BSM pseudorandom generator $\PRG:\{0,1\}^a\times\{0,1\}^b\rightarrow\{0,1\}$, as in the definition of the SADA problem.

\begin{enumerate}[leftmargin=15pt,rightmargin=10pt,itemsep=1pt,topsep=3pt]

\item For every $p\in\{0,1\}^d$ sample $k_p\in\{0,1\}^b$ uniformly.

\item Sample $r_1\in\{0,1\}^{\nu}$ uniformly, and instantiate algorithm $\AAA$ with read-once access to bits of $r_1$. Here $\nu$ bounds the number of coin flips made by $\AAA$.

\item For every $p\in P$, feed the update $(p,k_p)$ to $\AAA$.

\item Sample $r_2\in\{0,1\}^{\nu}$ uniformly, and switch the read-once access of $\AAA$ to $r_2$. (The switch from $r_1$ to $r_2$ is done for convenience, so that after Step~3 we do not need to ``remember'' the position for the next coin from $r_1$.)

\item REPEAT $\ell\triangleq\frac{m-n}{(a+1)\cdot 2^d}$ times 

\begin{enumerate}
	\item Obtain the next query $q:\{0,1\}^d\rightarrow\{0,1\}$.
	
	\item\label{step:loopp} For every $p\in\{0,1\}^d$ do
	
	\begin{enumerate}
		\item Sample $\Gamma\in\{0,1\}^a$ uniformly.
		\item Feed $a$ updates (one by one) to $\AAA$ s.t.\ the concatenation of their first bits is $\Gamma$.
		\item\label{step:red} \red{If $p\in P$ then let $Y=\PRG(\Gamma,k_p)$. Otherwise sample $Y\in\{0,1\}$ uniformly.}
		\item\label{step:OTP} Feed to $\AAA$ an update whose first bit is $Y\oplus q(p)$.
	\end{enumerate}
	
	\item\label{step:ObtainAnswer} Obtain an answer $z$ from $\AAA$.
	\item Output $z$.	

\end{enumerate}

\end{enumerate}
\end{algorithm*}

\begin{lemma}\label{lem:natural}
Algorithm \texttt{AnswerQueriesOTP} is natural.
\end{lemma}

\begin{proof}[Proof sketch]
This follows from the fact that the value of the given queries outside of the input sample $P$ are completely ``hidden'' from algorithm $\AAA$ (namely, by the classic ``one-time pad'' encryption scheme), and by observing that the answer $z$ given by algorithm \texttt{AnswerQueriesOTP} on a query $q$ is determined by the state of algorithm $\AAA$ at the end of the corresponding iteration of Step~5. 
\end{proof}

We now argue that the modification we introduced (from \texttt{AnswerQueries} to \texttt{AnswerQueriesOTP}) has basically no effect on the execution, and hence, algorithm \texttt{AnswerQueriesOTP} is both natural and statistically-accurate. This will lead to a contradiction.

\begin{lemma}\label{lem:TVdistance}
Suppose that $\AAA$ has space complexity $w$. Denote $\ell=\frac{m-n}{(a+1)\cdot 2^d}$. If $\PRG$ is an $\varepsilon$-secure BSM pseudorandom generator against adversaries with storage $O(w+ \ell + b\cdot 2^d)$, then for every input database $P$ and every adversary $\mathbb{A}$, the outcome distributions of 
$\accgame_{n, \ell,\texttt{AnswerQueries}, \mathbb{A}}(P)$ and $\accgame_{n, \ell,\texttt{AnswerQueriesOTP}, \mathbb{A}}(P)$ are within statistical distance $2^d m \eps$.
\end{lemma}

\begin{proof}
Recall that the outcome of $\accgame_{n, \ell,\texttt{AnswerQueries}, \mathbb{A}}(P)$ is the transcript of the interaction\linebreak $(q_1,z_1,\dots,q_{\ell},z_{\ell})$, where $q_i$ are the queries given by $\mathbb{A}$, and where $z_i$ are the answers given by $\texttt{AnswerQueries}$. We need to show that the distributions of $(q_1,z_1,\dots,q_{\ell},z_{\ell})$ during the executions with $\texttt{AnswerQueries}$ and $\texttt{AnswerQueriesOTP}$ are close. Without loss of generality, we assume that $\mathbb{A}$ is deterministic (indeed, if the lemma holds for every deterministic $\mathbb{A}$ then it also holds for every randomized $\mathbb{A}$). Hence, the transcript $(q_1,z_1,\dots,q_{\ell},z_{\ell})$ is completely determined by the answers given by the mechanism. So we only need to show that $(z_1,\dots,z_{\ell})$ is distributed similarly during the two cases. 
Note that, as we are aiming for constant accuracy, we may assume that each answer $z_i$ is specified using a constant number of bits (otherwise we can alter algorithm $\AAA$ to make this true while essentially maintaining its utility guarantees).

Now, for every $g\in\{0,1,2,\dots,2^d\}$, let $\texttt{AnswerQueries}_g$ denote an algorithm similar to algorithm \texttt{AnswerQueries}, except that in Step~\ref{step:red}, we set $Y=\PRG(\Gamma,k_p)$ if $p\in P$ or if $p\geq g$, and otherwise we sample $Y\in\{0,1\}$ uniformly.
Observe that $\texttt{AnswerQueries}_0\equiv \texttt{AnswerQueries}$ and that $\texttt{AnswerQueries}_{2^d}\equiv \texttt{AnswerQueriesOTP}$. We now show that for every $g$ it hods that the statistical distance between $\accgame_{n, \ell,\texttt{AnswerQueries}_g, \mathbb{A}}(P)$ and $\accgame_{n, \ell,\texttt{AnswerQueries}_{g+1}, \mathbb{A}}(P)$ is at most $\eps m$, which proves the lemma (by the triangle inequality).

Fix an index $g\in\{0,1,\dots,2^d-1\}$. Let $\accgame_g^*$ be an algorithm that simulates the interaction between $\mathbb{A}$ and $\texttt{AnswerQueries}_{g}$ on the database $P$, except that during an iteration of Step~\ref{step:loopp} with $p=g$, algorithm $\accgame_g^*$ gets $\Gamma$ and $Y$ as input, where $\Gamma$ is sampled uniformly and where $Y$ is either sampled uniformly from $\{0,1\}$ or computed as $Y=\PRG(\Gamma,k)$ for some key $k$ sampled uniformly from $\{0,1\}^b$ (unknown to $\texttt{AnswerQueries}_{g}$). These two cases correspond to  $\accgame_{n, \ell,\texttt{AnswerQueries}_{g+1}, \mathbb{A}}(P)$ and  $\accgame_{n, \ell,\texttt{AnswerQueries}_{g}, \mathbb{A}}(P)$, respectively. 

Observe that $\accgame_g^*$ can be implemented with storage space at most $\hat{W}=O(w+ \ell + b\cdot 2^d)$, specifically, for storing the internal state of algorithm $\AAA$ (which is $w$ bits), storing all previous answers $z_1,z_2,\dots,z_i$ (which is $O(\ell)$ bits), and storing all the keys $k_p$ for $p\neq g$ (which takes at most $b\cdot 2^d$ bits). Note that, as we assume that $\mathbb{A}$ is deterministic, on every step we can compute the next query from the previously given answers. 

Now, when $\accgame_g^*$ is given truly random bits $Y$, then it can be viewed as an adversary acting in the ideal experiment for $\PRG$ (see Section~\ref{sec:BSM}), and when $\accgame_g^*$ is given pseudorandom bits then it can be viewed as an adversary acting in the real experiment. By Theorem~\ref{thm:Vadhan}, assuming that $\PRG$ is $\eps$-secure against adversaries with storage $\hat{W}$, then the distribution on the storage of $\accgame_g^*$ in the two cases is close up to statistical distance $\eps m$. The lemma now follows from the fact that the sequence of answers $(z_1,\dots,z_{\ell})$ is included in the storage of $\accgame_g^*$.
\end{proof}

Combining Lemma~\ref{lem:streamingAccuracy} (stating that \texttt{AnswerQueries} is statistically-accurate) with Lemma~\ref{lem:TVdistance} (stating that \texttt{AnswerQueries} and \texttt{AnswerQueriesOTP} are close) we get that \texttt{AnswerQueriesOTP} must also be statistically-accurate. Formally,

\begin{lemma}\label{lem:OTPAccuracy}
Suppose that $\AAA$ is $(\alpha,\beta)$-accurate for the SADA problem for streams of length $m$ using memory $w$, and suppose that $\PRG$ is an $\eps$-secure BSM pseudorandom generator against adversaries with storage $O(w+ \ell + b\cdot 2^d)$, where $\ell=\frac{m-n}{(a+1)\cdot2^d}$. Then for every $\beta',\eps>0$, algorithm \texttt{AnswerQueriesOTP} is $\left(\frac{\alpha}{1-\gamma}+\alpha',\beta+\beta'+2^d m \eps\right)$-statistically-accurate for $\ell$ queries where
$$
\alpha'=O\left( \sqrt{\frac{w+\ln(\frac{\ell}{\beta'})}{n}}\right).
$$
\end{lemma}

So, Lemmas~\ref{lem:natural} and~\ref{lem:OTPAccuracy} state that algorithm \texttt{AnswerQueriesOTP} is both natural and statistically-accurate. To obtain a contradiction to Theorem~\ref{thm:adaNegative}, we instantiate Lemma~\ref{lem:OTPAccuracy} with the pseudorandom generator from Theorem~\ref{thm:Vadhan}. We obtain the following result.

\begin{theorem}\label{thm:main}
For every $w$, there exists a streaming problem over domain of size $\poly(w)$ and stream length $O(w^5)$ that requires at least $w$ space to be solved in the adversarial setting to within constant accuracy (small enough), but can be solved in the oblivious setting using space $O(\log^2(w))$.
\end{theorem}

\begin{proof}
To contradict Theorem~\ref{thm:adaNegative}, we want the (natural) algorithm \texttt{AnswerQueriesOTP} to answer more than $n^2$ queries over a domain of size $\Omega(n)$. So we set $\ell=\frac{m-n}{(a+1)\cdot2^d}=\Omega(n^2)$ and $d=O(1)+\log n$. Note that with these settings we have $m=\Theta(n^3\cdot a)$.

By Lemma~\ref{lem:OTPAccuracy}, in order to ensure that \texttt{AnswerQueriesOTP}'s answers are accurate (to within some small constant), we set $n=\Theta(w+\log(m))$ (large enough). We assume without loss of generality that $w\geq\log(m)$, as we can always increase the space complexity of $\AAA$. So $n=\Theta(w)$, and $m=\Theta(w^3\cdot a)$.

In addition, to apply Lemma~\ref{lem:OTPAccuracy}, we need to ensure that the conditions on the security of $\PRG$ hold. For a small constant $\tau>0$, we use the pseudorandom generator from Theorem~\ref{thm:Vadhan} with $\eps= \frac{\tau}{m\cdot2^d}=O(\frac{1}{mn})=O(\frac{1}{mw})$. To get security against adversaries with storage $O(w+ \ell + b\cdot 2^d)=O(w^2+bw)$, we need to ensure 
$$a=\Omega\left(w^2+bw\right) \qquad\text{and}\qquad b=\Omega\left(\log\left(\frac{a}{\eps}\right)\right)=\Theta(\log(am)).$$
It suffices to take $a=\Theta(w^2)$ and $b=\Theta(\log(wm))=\Theta(\log(w))$. 
Putting everything together, with these parameters, by Lemma~\ref{lem:OTPAccuracy}, we get that algorithm \texttt{AnswerQueriesOTP} answers $\ell=\Omega(n^2)$ adaptive queries over domain of size $\Omega(n)$, which contradicts Theorem~\ref{thm:adaNegative}. This means that an algorithm with space complexity $w$ cannot solve the $(a,b,d,m,n,\gamma)$-SADA problem to within (small enough) constant accuracy, where 
$a=\Theta(w^2)$, and $b=d=O(\log(w))$, and $m=\Theta(w^5)$, and $n=\Theta(w)$.

In contrast, by Observation~\ref{obs:obliviousSpace}, for constant $\alpha,\beta,\gamma$, the oblivious algorithm \texttt{ObliviousSADA} uses space $O(\log^2(w))$ in this settings.
\end{proof}

\section{A Computational Separation}\label{sec:SADA2}

In the previous sections we presented a streaming problem that can be solved in the oblivious setting using small space complexity, but requires large space complexity to be solved in the adversarial setting. Even though this provides a strong separation between adversarial streaming and oblivious streaming, a downside of our result is that the streaming problem we present (the SADA problem) is somewhat unnatural.

\begin{question}
Is there a ``natural'' streaming problem for which a similar separation holds?
\end{question}

In particular, one of the ``unnatural'' aspects of the SADA problem is that the target function depends on the {\em order} of the elements in the stream (i.e., it is an asymmetric function). 
Asymmetric functions can sometimes be considered ``natural'' in the streaming context (e.g., counting the number of inversions in a stream or finding the longest increasing subsequence). However, the majority of the ``classical'' streaming problems are defined by symmetric functions (e.g., counting the number of distinct elements in the stream or the number of heavy hitters).

\begin{question}\label{question:open}
Is there a {\em symmetric} streaming problem that can be solved using polylogarithmic space (in the domain size and the stream length) in the oblivious setting, but requires polynomial space in the adversarial setting?
\end{question}

In this section we provide a positive answer to this question for {\em computationally efficient} streaming algorithms. That is, unlike our separation from the previous sections (for the SADA problem) which is information theoretic, the separation we present in this section (for a symmetric target function) is computational. We consider Question~\ref{question:open} (its information theoretic variant) to be an important question for future work.

\subsection{The SADA2 Problem}

Let $\kappa\in\N$ be a security parameter, let $m\in\N$ denote the length of the stream, and let $d\in\N$ and $\gamma\in(0,1)$ be additional parameters. Let $(\Gen,\Enc,\Dec)$ be a semantically secure private-key encryption scheme, with key length $\kappa$ and ciphertext length $\psi=\poly(\kappa)$ for encrypting a message in $\{0,1\}$. We consider a streaming problem over a domain $X=\{0,1\}^{1+d+\log(m)+\psi}$, where an update $x\in X$ has two possible types (the type is determined by the first bit of $x$):

\paragraph{Data update:} $x=(0,p,k)\in \{0,1\}\times\{0,1\}^d\times\{0,1\}^{\kappa}$,
\paragraph{Query update:} $x=(1,p,j,c)\in \{0,1\}\times\{0,1\}^d\times\{0,1\}^{\log m}\times\{0,1\}^{\psi}$.\\

We define a function $g:X^*\rightarrow[0,1]$ as follows. Let $\vec{x}=\{x_1,\dots,x_i\}$ be a sequence of updates.
For $p\in\{0,1\}^d$, let $x_{i_1}{=}(0,p,k_{i_1}),\dots,x_{i_{\ell}}{=}(0,p,k_{i_{\ell}})$ denote all the ``data updates'' in $\vec{x}$ with the point $p$, and let $k_{i_1},\dots,k_{i_{\ell}}$ denote their corresponding keys (some of which may be identical). Now let $k_{p}=k_{i_1}\wedge\dots\wedge k_{i_{\ell}}$. That is, $k_p$ is the bit-by-bit AND of all of the keys that correspond to ``data updates'' with the point $p$. 
Now let $S$ be the set that %
contains the pair $(p,k_p)$ for every $p$ such that there exists a ``data update'' in $\vec{x}$ with the point $p$. Importantly, $S$ is a {\em set} rather than a multiset.
Similarly to the previous sections, we also add special symbols, $\bot_1,\dots,\bot_{\gamma 2^d}$, to $S$. %
Formally, $S$ is constructed as follows.
\begin{center}
\noindent\fboxother{
\parbox{.9\columnwidth}{
\begin{enumerate}
	\item Initiate $S=\{\bot_1,\dots,\bot_{\gamma 2^d}\}$.
	\item For every $p\in\{0,1\}^d$:
	\begin{enumerate}
		\item Let $x_{i_1}{=}(0,p,k_{i_1}),\dots,x_{i_{\ell}}{=}(0,p,k_{i_{\ell}})$ denote all the ``data updates'' in $\vec{x}$ (i.e., updates beginning with 0) that contain the point $p$.
		\item If $\ell>0$ then let $k_{p}=k_{i_1}\wedge\dots\wedge k_{i_{\ell}}$ and add $(p,k_p)$ to $S$.
	\end{enumerate}
\end{enumerate}
}}
\end{center}

\medskip
We now define the query $q$ that corresponds to $\vec{x}$. First, $q(\bot_1)=\dots=q(\bot_{\gamma 2^d})=1$. Now, for $p\in\{0,1\}^d$, let 
$$
j_{p}^{\rm max}=\max\left\{  j : \exists c\in\{0,1\}^{\psi} \text{ such that }  (1,p,j,c)\in\vec{x} \right\}.
$$
That is, $j_{p}^{\rm max}$ denotes the maximal index such that $(1,p,j_{p}^{\rm max},c)$ appears in $\vec{x}$ for some $c\in\{0,1\}^{\psi}$. Furthermore, let $x_{i_1}{=}(1,p,j_{p}^{\rm max},c_{i_1}),\dots,x_{i_{\ell}}{=}(1,p,j_{p}^{\rm max},c_{i_{\ell}})$ denote the ``query updates'' with $p$ and $j_{p}^{\rm max}$. %
Now let $c_p=c_{i_1}\wedge\dots\wedge c_{i_{\ell}}$.
That is, $c_p$ is the bit-by-bit AND of all of the ciphertexts that correspond to ``query updates'' with $p$ and $j_{p}^{\rm max}$. 
If the point $p$ does not appear in any ``query update'' then we set $c_p=\vec{1}$ by default. 
The query $q:\left(\{0,1\}^d\times\{0,1\}^{\kappa}\right)\rightarrow\{0,1\}$ is defined as $q(p,k) = \Dec(c_p , k )$.

\medskip
Finally, the value of the function $g$ on the stream $\vec{x}$ is defined to be
$$
g(\vec{x}) = q(S) = \frac{1}{|S|}\left[ \gamma 2^d +  \sum_{(p,k_p)\in S} q(p , k_p )\right].
$$ 
That is, $g(\vec{x})$ returns the average of $q$ on $S$. %

\begin{definition}[The $(d,m,\kappa,\gamma)$-SADA2 Problem]
At every time step $i\in[m]$, after obtaining the next update $x_i\in X$, the goal is to approximate $g(x_1,\dots,x_i)$.
\end{definition}

\subsection{An Oblivious Algorithm for the SADA2 Problem}

\begin{algorithm*}[t!]
\caption{\bf \texttt{ObliviousSADA2}}\label{alg:ObliviousSADA2}

{\bf Setting:} On every time step we obtain the next update, which is an element of $X=\{0,1\}^{1+d+\log(m)+\psi}$.

\begin{enumerate}[leftmargin=15pt,rightmargin=10pt,itemsep=1pt,topsep=3pt]

\item Let $D$ be a sample (multiset) containing $O(\frac{1}{\alpha^2\gamma^2}\ln(\frac{m}{\beta}))$ i.i.d.\ elements chosen uniformly from $\{0,1\}^d\cup\{\bot_1,\dots,\bot_{\gamma 2^d}\}$, and let $D_{\bot}\leftarrow D\cap\{\bot_1,\dots,\bot_{\gamma 2^d}\}$, and let $D_{X}\leftarrow D\setminus D_{\bot}$.
\item For every $p\in D_X$, let $\texttt{inS}_p\leftarrow 0$, let $k_p\leftarrow \vec{1}$, let $j_p\leftarrow 0$, and let $c_p\leftarrow \vec{1}$.

\item REPEAT

\begin{enumerate}

	\item Obtain the next update in the stream $x$.
	\item If the first bit of $x$ is 0 then
	\begin{enumerate}
		\item Denote $x=(0,p,k)$.
		\item If $p\in D_X$ then let $\texttt{inS}_p\leftarrow 1$ and let $k_p\leftarrow k_p\wedge k$.
		
	\end{enumerate}

	\item If the first bit of $x$ is 1 then
	\begin{enumerate}
		\item Denote $x=(1,p,j,c)$.
		\item If $p\in D_X$ and $j=j_p$ then set $c_p\leftarrow c_p\wedge c$. 
		\item If $p\in D_X$ and $j>j_p$ then set $c_p\leftarrow c$ and $j_p\leftarrow j$.		
	\end{enumerate}
	
	\item Let $v\leftarrow \left|\{p\in D_X : \texttt{inS}_p=1 \}\right|+\left|D_{\bot}\right|$ and let $z\leftarrow \frac{\left|D_{\bot}\right|}{v}$.
	
	\item For every $p\in D_X$ such that $\texttt{inS}_p=1$ set $z\leftarrow z + \frac{\Dec(c_p,k_p)}{v}$.

	\item\label{step:sada2out} Output $z$.
	\end{enumerate}

\end{enumerate}

\end{algorithm*}

In this section we present an oblivious streaming algorithm for the SADA2 problem. 
The algorithm begins by sampling a multiset $D$ containing a small number of random elements from the domain $\{0,1\}^d\cup\{\bot_1,\dots,\bot_{\gamma 2^d}\}$. The algorithm then proceeds by maintaining the set $S$ and the query $q$ (which are determined by the input stream; as in the definition of the SADA2 problem) only w.r.t.\ elements that appear in the sample $D$. As we next explain, in the oblivious setting, this suffices in order to accurately solve the SADA2 problem. Consider algorithm \texttt{ObliviousSADA2}, given in Algorithm~\ref{alg:ObliviousSADA2}.

\begin{theorem}
Assume that
$2^d=\Omega(\frac{1}{\gamma}\ln(\frac{m}{\beta}))$ and $|D|\geq\Omega(\frac{1}{\alpha^2\gamma^2}\ln(\frac{m}{\beta}))$. Then \texttt{ObliviousSADA2} is $(\alpha,\beta)$-accurate for the SADA2 problem in the oblivious setting.
\end{theorem}

\begin{proof}
Fix the stream $\vec{x}_m=(x_1,\dots,x_m)$. Fix a time step $i\in[m]$, and consider the prefix $\vec{x}_i=(x_1,\dots,x_i)$. 
Let $S_i=S_i(\vec{x}_i)$ be the {\em set} and let $q_i=q_i(\vec{x}_i)$ be the {\em query} defined by $\vec{x}_i$, as in the definition of the SADA2 problem. 
Consider the multiset $T=\{ (p,k_p)  : p\in D_X \text{ and } \texttt{inS}_p=1  \}\cup D_{\bot}$.
Let $z_i$ be the answer returned in Step~\ref{step:sada2out} after precessing the update $x_i$. Observe that $z_i$ is exactly the average of $q_i$ on the multiset $T$, that is, $z_i=q_i(T)$.

Recall that $|S_i|\geq\gamma 2^d$, and recall that every element in $D$ is sampled uniformly from $\{0,1\}^d \cup \{\bot_1,\dots,\bot_{\gamma 2^d}\}$. Therefore, $\E_D[|D\cap S_i|]\geq|D|\cdot \frac{\gamma 2^d}{2^d+\gamma2^d}=|D|\cdot\frac{\gamma}{1+\gamma}$. By the Chernoff bound, assuming that $2^d=\Omega(\frac{1}{\gamma}\ln(\frac{m}{\beta}))$, then with probability at least $1-\frac{\beta}{m}$ we have that $|D\cap S_i|\geq\frac{\gamma}{2}|D|$. We proceed with the analysis assuming that this is the case.

Now, for every $t\geq\frac{\gamma}{2}|D|$, when conditioning on $|D\cap S_i|=t$ we have that $T$ is a sample containing $t$ i.i.d.\ elements from $S_i$. In that case, again using the Chernoff bound, with probability at least $1-\frac{\beta}{m}$ we have that $z_i=q_i(T)\in(1\pm\alpha)\cdot q_i(S_i)$, assuming that $t\geq\Omega(\frac{1}{\alpha^2\gamma}\ln(\frac{m}{\beta}))$. This assumption holds when $|D|\geq\Omega(\frac{1}{\alpha^2\gamma^2}\ln(\frac{m}{\beta}))$.

So, for every fixed $i$, with probability at least $1-O(\frac{\beta}{m})$ we have that $z_i\in(1\pm\alpha)\cdot q_i(S_i)$. By a union bound, this holds for every time step $i$ with probability at least $1-O(\beta)$.
\end{proof}

\begin{observation}\label{obs:oblivious2Space}
For constant $\alpha,\beta,\gamma$, algorithm \texttt{ObliviousSADA2} uses space $\tilde{O}\left( \log(m)\cdot\log|X|\right)$, in addition to the space required by $\Dec$.
\end{observation}

\subsection{A Negative Result for the SADA2 Problem}

We now show that the SADA2 problem cannot be solved efficiently in the adversarial setting. To that end, suppose we have an adversarially robust streaming algorithm $\AAA$ for the SADA2 problem, and consider algorithm \texttt{AnswerQueries2} that uses $\AAA$ in order to solve the ADA problem. Recall that in the SADA2 problem the collection of ``data updates'' is treated as a {\em set}, while the input to an algorithm for the ADA problem is a {\em multiset}. 
In the following claim we show that \texttt{AnswerQueries2} is empirically-accurate w.r.t.\ its input (when treated as a set). 

\begin{algorithm*}[t!]
\caption{\bf \texttt{AnswerQueries2}}\label{alg:AnswerQueries2}

{\bf Input:} A database $P$ containing $n$ elements from $\{0,1\}^d$.

{\bf Setting:} On every time step we get a query $q:\{0,1\}^d\rightarrow\{0,1\}$.

{\bf Algorithm used:} An adversarially robust streaming algorithm $\AAA$ for the $(d,m,\kappa,\gamma)$-SADA2 problem with $(\alpha,\beta)$-accuracy for streams of length $m$. We abstract the coin tosses of $\AAA$ using {\em two} random strings, $r_1$ and $r_2$, of possibly unbounded length. Initially, we execute $\AAA$ with access to $r_1$, meaning that every time it tosses a coin it gets the next bit in $r_1$. At some point, we switch the random string to $r_2$, and henceforth $\AAA$ gets its coin tosses from $r_2$.

{\bf Algorithm used:} Encryption scheme $(\Gen,\Enc,\Dec)$, as in the definition of the SADA2 problem.

\begin{enumerate}[leftmargin=15pt,rightmargin=10pt,itemsep=1pt,topsep=3pt]

\item For every $p\in\{0,1\}^d$ sample $k_p\leftarrow\Gen(1^{\kappa})$ independently.

\item Sample $r_1\in\{0,1\}^{\nu}$ uniformly, and instantiate algorithm $\AAA$ with read-once access to bits of $r_1$. Here $\nu$ bounds the number of coin flips made by $\AAA$.

\item For every $p\in P$, feed the update $(0,p,k_p)$ to $\AAA$.

\item Sample $r_2\in\{0,1\}^{\nu}$ uniformly, and switch the read-once access of $\AAA$ to $r_2$. (The switch from $r_1$ to $r_2$ is done for convenience, so that after Step~3 we do not need to ``remember'' the position for the next coin from $r_1$.)

\item For $j=1$ to $\ell \triangleq\frac{m-n}{2^d}$ do

\begin{enumerate}
	\item Obtain the next query $q_j:\{0,1\}^d\rightarrow\{0,1\}$.
	
	\item For every $p\in\{0,1\}^d$ do
	
	\begin{enumerate}
		\item Let $c_p=\Enc(q_j(p) , k_p)$.
		\item Feed the update $(1,p,j,c_p)$ to $\AAA$.
	\end{enumerate}
	
	\item Obtain an answer $z$ from $\AAA$.
	\item Output $z$.
	
\end{enumerate}

\end{enumerate}
\end{algorithm*}

\begin{claim}%
Let $P\in\left(\{0,1\}^d\right)^*$ be an input multiset, let $\tilde{P}$ be the {\em set} containing every point that appears in $P$, and assume that $|\tilde{P}|=n$. 
If $\AAA$ is $(\alpha,\beta)$-accurate for the SADA2 problem, then $\texttt{AnswerQueries2}\left(P\right)$ is $\left(\alpha+\frac{\gamma\cdot2^d}{n},\beta\right)$-empirically-accurate for $\frac{m-n}{2^d}$ adaptively chosen statistical queries w.r.t.\ the set $\tilde{P}$.
\end{claim}

\begin{proof}[Proof sketch]
Let $q$ denote the query given at some iteration, and let $q_{\vec{x}}$ and $S_{\vec{x}}$ denote the query and the dataset specified by the updates given to algorithm $\AAA$. The claim follows from the fact that, by construction, for every $p\in \tilde{P}$ we have that $q(p)=q_{\vec{x}}(p,k_p)$. Therefore, 
$$
q_{\vec{x}}(S_{\vec{x}})=\frac{1}{\left|S_{\vec{x}}\right|}\left[ \gamma\cdot2^d + \sum_{p\in \tilde{P}} q(p) \right]
=\frac{1}{n+\gamma2^d}\left[ \gamma\cdot2^d + \sum_{p\in \tilde{P}} q(p) \right]
=\frac{ \frac{\gamma\cdot2^d}{n}  }{1+\frac{\gamma\cdot2^d}{n}} +   \frac{1}{n+\gamma2^d} \sum_{p\in \tilde{P}} q(p). 
$$
Therefore, $q_{\vec{x}}(S_{\vec{x}})\leq \frac{\gamma\cdot2^d}{n} + q(\tilde{P})$, and also $q_{\vec{x}}(S_{\vec{x}})\geq  \frac{1}{n+\gamma2^d} \sum_{p\in \tilde{P}} q(p)$ which means that $q(\tilde{P})\leq \frac{n+\gamma2^d}{n}\cdot q_{\vec{x}}(S_{\vec{x}})\leq q_{\vec{x}}(S_{\vec{x}}) + \frac{\gamma2^d}{n}$. So, whenever the answers given by $\AAA$ are $\alpha$-accurate w.r.t.\ $q_{\vec{x}}(S_{\vec{x}})$, they are also $\left(\alpha+\frac{\gamma2^d}{n}\right)$-accurate w.r.t.\ $\tilde{P}$.
\end{proof}

We now show that algorithm \texttt{AnswerQueries2} is transcript-compressible. To that end, for every choice of $\vec{k},r_1,r_2,\vec{r}_{\Enc}$ for the keys $k$, the random bitstrings $r_1,r_2$, and the randomness used by $\Enc$ at its different executions, let us denote by $\texttt{AnswerQueries2}_{\vec{k},r_1,r_2,\vec{r}_{\Enc}}$ algorithm \texttt{AnswerQueries2} after fixing these elements.

\begin{claim}%
If algorithm $\AAA$ uses space at most $w$, then, for every $\vec{k},r_1,r_2$, we have that algorithm $\texttt{AnswerQueries2}_{\vec{k},r_1,r_2,\vec{r}_{\Enc}}$ is transcript-compressible to $w$ bits.
\end{claim}

\begin{proof}[Proof sketch]
Assuming that the adversary who generates the queries $q$ is deterministic (which is without loss of generality) we get that the entire transcript is determined by the state of algorithm $\AAA$ at the end of Step 3.
\end{proof}

Similarly to our arguments from Section~\ref{sec:negative}, since algorithm \texttt{AnswerQueries2} is both empirically-accurate and transcript-compressible, we get that it is also statistically-accurate. 
Since we only argued empirical-accuracy when treating the input multiset as a set, we will only argue for statistical-accuracy w.r.t.\ the uniform distribution, where we have that the difference between a random set and a random multiset is small. 
Formally,

\begin{lemma}%
Suppose that $\AAA$ is $(\alpha,\beta)$-accurate for the SADA2 problem for streams of length $m$ using memory $w$. Then for every $\beta'>0$, algorithm \texttt{AnswerQueries2} is $\left(\tilde{\alpha},\tilde{\beta}\right)$-statistically-accurate for $\ell=\frac{m-n}{2^d}$ queries w.r.t.\ the uniform distribution over $\{0,1\}^d$, where
$$
\tilde{\alpha}=
O\left( \alpha+\frac{\gamma\cdot2^d}{n} + \frac{n}{2^d}+  \sqrt{\frac{w+\ln(\frac{\ell}{\beta'})}{n}}\right)
\qquad\text{and}\qquad
\tilde{\beta}=O\left( \beta+\beta'+\exp\left(-\frac{n^2}{3\cdot 2^d}\right)  \right).
$$
\end{lemma}

\begin{proof}[Proof sketch]
The proof is analogous to the proof of Lemma~\ref{lem:streamingAccuracy}, with the following addition. 
Let $P$ be a multiset containing $n$ i.i.d.\ uniform samples from $\{0,1\}^d$, and let $\tilde{P}$ be the set containing every element of $P$. As we are considering the uniform distribution on $\{0,1\}^d$, then by the Chernoff bound, with probability at least $1-\exp(-\frac{n^2}{3\cdot 2^d})$, it holds that the set $\tilde{P}$ and and the multiset $P$ differ by at most $\frac{n^2}{2\cdot2^d}$ points, i.e., by at most an $\frac{n}{2\cdot2^d}$-fraction of the points. In that case, for every query $q$ we have that $|q(P)-q(\tilde{P})|\leq\frac{n}{2\cdot2^d}$.
\end{proof}

So algorithm \texttt{AnswerQueries2} is statistically-accurate. 
To obtain a contradiction, we modify the algorithm such that it becomes natural. 
Consider algorithm \texttt{AnswerQueries2Natural}. 
As before, the modifications are marked in red.

\begin{algorithm*}[t!]
\caption{\bf \texttt{AnswerQueries2Natural}}\label{alg:AnswerQueries2Natural}

{\bf Input:} A database $P$ containing $n$ elements from $\{0,1\}^d$.

{\bf Setting:} On every time step we get a query $q:\{0,1\}^d\rightarrow\{0,1\}$.

{\bf Algorithm used:} An adversarially robust streaming algorithm $\AAA$ for the $(d,m,\kappa,\gamma)$-SADA2 problem with $(\alpha,\beta)$-accuracy for streams of length $m$. We abstract the coin tosses of $\AAA$ using {\em two} random strings, $r_1$ and $r_2$, of possibly unbounded length. Initially, we execute $\AAA$ with access to $r_1$, meaning that every time it tosses a coin it gets the next bit in $r_1$. At some point, we switch the random string to $r_2$, and henceforth $\AAA$ gets its coin tosses from $r_2$.

{\bf Algorithm used:} Encryption scheme $(\Gen,\Enc,\Dec)$, as in the definition of the SADA2 problem.

\begin{enumerate}[leftmargin=15pt,rightmargin=10pt,itemsep=1pt,topsep=3pt]

\item For every $p\in\{0,1\}^d$ sample $k_p\leftarrow\Gen(1^{\kappa})$ independently.

\item Sample $r_1\in\{0,1\}^{\nu}$ uniformly, and instantiate algorithm $\AAA$ with read-once access to bits of $r_1$. Here $\nu$ bounds the number of coin flips made by $\AAA$.

\item For every $p\in P$, feed the update $(0,p,k_p)$ to $\AAA$.

\item Sample $r_2\in\{0,1\}^{\nu}$ uniformly, and switch the read-once access of $\AAA$ to $r_2$. (The switch from $r_1$ to $r_2$ is done for convenience, so that after Step~3 we do not need to ``remember'' the position for the next coin from $r_1$.)

\item For $j=1$ to $\ell \triangleq\frac{m-n}{2^d}$ do

\begin{enumerate}
	\item Obtain the next query $q_j:\{0,1\}^d\rightarrow\{0,1\}$.
	
	\item For every $p\in\{0,1\}^d$ do
	
	\begin{enumerate}
		\item \red{If $p\in P$ then let $c_p=\Enc(q_j(p) , k_p)$. Otherwise let $c_p=\Enc(0 , k_p)$.}
		\item Feed the update $(1,p,j,c_p)$ to $\AAA$.
	\end{enumerate}
	
	\item Obtain an answer $z$ from $\AAA$.
	\item Output $z$.
	
\end{enumerate}

\end{enumerate}
\end{algorithm*}

\begin{observation}\label{obs:natural2}
Algorithm \texttt{AnswerQueries2Natural} is natural.
\end{observation}

\begin{proof}[Proof sketch]
This follows from the fact that the value of the given queries outside of the input sample $P$ are ignored, and are replaced with (encryptions of) zero.  
\end{proof}

The following lemma follows from the assumed security of the encryption scheme.

\begin{lemma}\label{lem:computational}
Suppose that $(\Gen,\Enc,\Dec)$ is semantically secure private-key encryption scheme with key length $\kappa=\kappa(m)$ against adversaries with time $\poly(m)$. 
Fix $\alpha\in(0,1)$. 
Let $\mathbb{A}$ be a data analyst with running time $\poly(m)$. For a mechanism $\cM$ that answers queries, consider the interaction between $\cM$ and $\mathbb{A}$ , and let $E$ denote the event that $\cM$ failed to be $\alpha$-statistically accurate at some point during the interaction. Then, for an input database $P$ sampled uniformly from $\{0,1\}^d$ it holds that
$$
\left|
\Pr_{P,\mathbb{A},\texttt{AnswerQueries2}(P)}[E]
-
\Pr_{P,\mathbb{A},\texttt{AnswerQueries2Natural}(P)}[E]
\right|\leq\negl(\kappa).
$$
\end{lemma}

The proof of Lemma~\ref{lem:computational} is straightforward from the definition of security. We give here the details for completeness. To that end, let us recall the formal definition of security of an encryption scheme. Consider a pair of oracles $\oracle_0$ and $\oracle_1$, where $\oracle_1(k_1,\dots,k_N,\cdot)$ takes as input an index of a key $i\in[N]$ and a message $M$ and returns $\Enc(M,k_i)$, and where $\oracle_0(k_1,\dots,k_N,\cdot)$ takes the same input
but returns $\Enc(0,k_i)$. 
An encryption scheme $(\Gen,\Enc,\Dec)$ is {\em secure} if no computationally efficient adversary can tell whether it is interacting with $\oracle_0$ or with $\oracle_1$. Formally,

\begin{definition}
Let $m:\R\rightarrow\R$ be a %
function. 
An encryption scheme $(\Gen,\Enc,\Dec)$ is $m$-\emph{secure} if for every $N = \poly(m(\kappa))$, and every $\poly(m(\kappa))$-time adversary $\BBB$, the following holds.
\begin{equation*}
\left| 
\Pr_{\substack{k_1,\dots,k_N \\ \BBB,\Enc }}\left[  \BBB^{\oracle_0(k_1,\dots,k_N,\cdot)}=1  \right] 
- 
\Pr_{\substack{k_1,\dots,k_N \\ \BBB,\Enc }}\left[  \BBB^{\oracle_1(k_1,\dots,k_N,\cdot)}=1  \right] 
\right| = \negl(\kappa),
\end{equation*}
where the probabilities are over sampling $k_1,\dots,k_N\leftarrow\Gen(1^{\kappa})$ and over the randomness of $\BBB$ and $\Enc$.
\end{definition}

\begin{remark}
When $m$ is the identity function we simply say that $(\Gen,\Enc,\Dec)$ is {\em secure}. Note that in this case, security holds against all adversaries with runtime polynomial in the security parameter $\kappa$. We will further assume the existence of a {\em sub-exponentially secure} encryption scheme. By that we mean that there exist a constant $\tau>0$ such that $(\Gen,\Enc,\Dec)$ is $m$-secure for $m(\kappa)=2^{\kappa^\tau}$. That is, we assume the existence of an encryption scheme in which security holds agains all adversaries with runtime polynomial in $2^{\kappa^\tau}$.
\end{remark}

To prove Lemma~\ref{lem:computational} we construct an adversary $\BBB$ for the scheme $(\Gen,\Enc,\Dec)$ such that its advantage in breaking the the security of $(\Gen,\Enc,\Dec)$ is exactly the difference in the probability of event $E$ between the execution with \texttt{AnswerQueries2} or with \texttt{AnswerQueries2Natural}. This implies that the difference between these two probabilities is negligible.

\begin{algorithm*}[t!]
\caption{\bf An adversary ${\boldsymbol{\BBB}}$ for the encryption scheme}\label{alg:B}

{\bf Algorithm used:} An adversarially robust streaming algorithm $\AAA$ for the $(d,m,\kappa,\gamma)$-SADA2 problem with $(\alpha,\beta)$-accuracy for streams of length $m$. 

{\bf Algorithm used:} A data analyst $\mathbb{A}$ that outputs queries and obtains answers. 

{\bf Algorithm used:} Encryption scheme $(\Gen,\Enc,\Dec)$.

{\bf Oracle access:} $\oracle_b(k_1,\dots,k_N,\cdot)$ where $b\in\{0,1\}$ and where $N=2^d$ and $k_1,\dots,k_N\leftarrow\Gen(1^{\kappa})$. 

\begin{enumerate}[leftmargin=15pt,rightmargin=10pt,itemsep=1pt,topsep=3pt]

\item Let $P$ be a multiset containing $n$ uniform samples from $\{0,1\}^d$.

\item For every $p\in P$ sample $\bar{k}_p\leftarrow\Gen(1^{\kappa})$ independently.

\item Instantiate algorithm $\AAA$.

\item For every $p\in P$, feed the update $(0,p,\bar{k}_p)$ to $\AAA$.

\item Instantiate the data analyst $\mathbb{A}$.

\item For $j=1$ to $\ell \triangleq\frac{m-n}{2^d}$ do

\begin{enumerate}
	\item Obtain the next query $q_j:\{0,1\}^d\rightarrow\{0,1\}$ from the data analyst $\mathbb{A}$.
	
	\item For every $p\in\{0,1\}^d$ do
	
	\begin{enumerate}
		\item If $p\in P$ then let $c_p=\Enc(q_j(p) , \bar{k}_p)$. Otherwise let $c_p\leftarrow \oracle_b(p,q_j(p))$.
		\item Feed the update $(1,p,j,c_p)$ to $\AAA$.
	\end{enumerate}
	
	\item Obtain an answer $z$ from $\AAA$, and give $z$ to $\mathbb{A}$.
\end{enumerate}

\item Output 1 if and only if event $E$ occurs.
\end{enumerate}
\end{algorithm*}

\begin{proof}[Proof of Lemma~\ref{lem:computational}]
Let $\mathbb{A}$ be a data analyst with running time $\poly(m)$, and consider algorithm $\BBB$. First observe that if $\AAA$ and $\mathbb{A}$ are computationally efficient (run in time $\poly(m)$) then so is algorithm $\BBB$.

Now observe that when the oracle is $\oracle_1$ and when $k_1,\dots,k_N$ are chosen randomly from $\Gen(1^{\kappa})$ then  $\BBB^{\oracle_1(k_1,\dots,k_N,\cdot)}$ simulates the interaction between $\mathbb{A}$ and \texttt{AnswerQueries2} on a uniformly sampled database $P$. Similarly, when   the oracle is $\oracle_0$ and when $k_1,\dots,k_N$ are chosen randomly from $\Gen(1^{\kappa})$ then $\BBB^{\oracle_0(k_1,\dots,k_N,\cdot)}$ simulates the interaction between $\mathbb{A}$ and \texttt{AnswerQueries2Natural} on a uniformly sampled database $P$. Thus,
\begin{align*}
&\left|
\Pr_{P,\mathbb{A},\texttt{AnswerQueries2}(P)}[E]
-
\Pr_{P,\mathbb{A},\texttt{AnswerQueries2Natural}(P)}[E]
\right|\\[1em]
&=\left| 
\Pr_{\substack{k_1,\dots,k_N \\ \BBB,\Enc }}\left[  \BBB^{\oracle_1(k_1,\dots,k_N,\cdot)}=1  \right] 
- 
\Pr_{\substack{k_1,\dots,k_N \\ \BBB,\Enc }}\left[  \BBB^{\oracle_0(k_1,\dots,k_N,\cdot)}=1  \right] 
\right| = \negl(\kappa).
\end{align*}
\end{proof}

So, algorithm \texttt{AnswerQueries2Natural} is natural, and when $\AAA$ and $\mathbb{A}$ are computationally efficient, then the probability that \texttt{AnswerQueries2Natural} fails to be statistically-accurate is similar to the probability that \texttt{AnswerQueries2} fails, which is small. We therefore get the following lemma.

\begin{lemma}\label{lem:AQ2Nfinal}
Algorithm \texttt{AnswerQueries2Natural} is natural. In addition, 
if $(\Gen,\Enc,\Dec)$ is an $m$-secure private-key encryption scheme with key length $\kappa=\kappa(m)$, and if $\AAA$ is an adversarially robust streaming algorithm for the $(d,m,\kappa,\gamma)$-SADA2 problem with space $w$ and runtime $\poly(m)$, then \texttt{AnswerQueries2Natural} is
$\left(\tilde{\alpha},\tilde{\beta}\right)$-statistically-accurate for $\ell=\frac{m-n}{2^d}$ queries w.r.t.\ the uniform distribution over $\{0,1\}^d$, and w.r.t.\ a data analyst $\mathbb{A}$ with running time $\poly(m)$, where
$$
\tilde{\alpha}=
O\left( \alpha+\frac{\gamma\cdot2^d}{n} + \frac{n}{2^d}+  \sqrt{\frac{w+\ln(\frac{\ell}{\beta'})}{n}}\right)
\qquad\text{and}\qquad
\tilde{\beta}=O\left( \beta+\beta'+\exp\left(-\frac{n^2}{3\cdot 2^d}\right) +\negl(\kappa) \right).
$$
\end{lemma}

We now restate Theorem~\ref{thm:adaNegative}, in which we simplified the results of Steinke and Ullman. In this section we use the stronger formulation of their results, given as follows.

\begin{theorem}[\cite{SU15}]\label{thm:adaNegative2}
There exists a constant $c>0$ such that no natural algorithm is $(c,c)$-statistically-accurate for $O(n^2)$ adaptively chosen queries given $n$ samples over a domain of size $\Omega(n)$. Furthermore, this holds even when assuming that the data analyst is computationally efficient (runs in time $\poly(n^2)$) and even when the underlying distribution is the uniform distribution.
\end{theorem}

Combining Lemma~\ref{lem:AQ2Nfinal} with Theorem~\ref{thm:adaNegative2} we obtain the following result.

\begin{theorem}
Assume the existence of a sub-exponentially secure private-key encryption scheme. 
Then, the $\left(d{=}\Theta(\log m),m,\kappa{=}\polylog(m),\gamma{=}\Theta(1)\right)$-SADA2 problem can be solved in the oblivious setting to within constant accuracy using space $\polylog(m)$ and using $\polylog(m)$ runtime (per update). 
In contrast, every adversarially robust algorithm for this problem with $\poly(m)$ runtime per update must use space $\poly(m)$.
\end{theorem}

\bibliographystyle{abbrv}

\end{document}